\documentclass[a4paper,twoside,11pt,reqno]{amsart}

\usepackage{amsmath,amsfonts,amsbsy,amsgen,amscd,mathrsfs,amssymb,amsthm,bm}
\usepackage{enumerate,mathtools}
\usepackage[margin=1.25in]{geometry}

\usepackage[usenames,dvipsnames]{xcolor}
\usepackage[colorlinks=true,citecolor=blue,linkcolor=blue]{hyperref}
\usepackage{comment}

\usepackage{tikz}
\usetikzlibrary{shadings}

\newtheorem{theorem}{Theorem}[section]
\newtheorem{lemma}[theorem]{Lemma}
\newtheorem{corollary}[theorem]{Corollary}
\newtheorem*{defn}{Definition}
\newtheorem*{prop*}{Proposition}

\newtheorem*{conj*}{Conjecture}

\newtheorem*{fact*}{Fact}
\newtheorem{prop}[theorem]{Proposition}
\newtheorem*{ex*}{Example}

\newtheorem{rem}[theorem]{Remark}

\theoremstyle{definition}

\newcommand*{\claimproofname}{Proof of claim}

\DeclareMathOperator*{\E}{\mathbb{E}}

\newcommand{\eps}{\varepsilon}

\DeclareMathSymbol{\shortminus}{\mathbin}{AMSa}{"39}

\newcommand{\Inf}{\mathbf{Inf}}

\newcommand{\bsa}{\mathbf{BSA}}
\newcommand{\sens}[1]{s_{#1}}
\newcommand{\maxbsa}{\mathbf{MBSA}}
\newcommand{\maxbsar}{\mathbf{MRBSA}}
\newcommand{\maxas}{\mathbf{MAS}}
\newcommand{\maxasr}{\mathbf{MRAS}}
\newcommand{\as}{\mathbf{AS}}
\newcommand{\G}{\Pi}

\numberwithin{equation}{section}
\numberwithin{figure}{section}
\numberwithin{table}{section}

\newcommand{\bE}{\mathbb{E}}

\newcommand{\R}{\mathbf{R}}

\newcommand{\bP}{\mathbf{P}}

\renewcommand{\gg}{>\!\!>}
\DeclareMathOperator{\sgn}{\operatorname{sgn}}

\newcommand{\hg}{\textnormal{Hg}}

\newcommand{\var}{\textnormal{Var}}

\allowdisplaybreaks

\begin{document}

\title[Boolean surface area of PTFs]{Estimates for the Boolean surface area of\\polynomial threshold functions}

\author[F. Chang]{Fan Chang}
\address{(F.C.) School of Statistics and Data Science, Nankai University, Tianjin, China and Extremal Combinatorics and Probability Group (ECOPRO), Institute for Basic Science (IBS), Daejeon, South Korea.}
\email{1120230060@mail.nankai.edu.cn.}

\author[J. Slote]{Joseph Slote}
\address{(J.S.) Department of Computer Science 
University of Washington, Seattle, 98195, USA }
\email{jslote@uw.edu}

\author[C-K. Tseng]{Chun-Kai Tseng}
\address{(C-K. T.) Department of Mathematics, MSU, 
East Lansing, MI 48823, USA}
\email{tsengc10@msu.edu}

\author[A. Volberg]{Alexander Volberg}
\address{(A.V.) Department of Mathematics, MSU, 
East Lansing, MI 48823, USA, and Max Planck Institute for Mathematics, Bonn, 5311, Germany }
\email{volberg@math.msu.edu}

\author[H. Zhang]{Haonan Zhang}
\address{(H.Z.) Department of Mathematics,
University of South Carolina,
Columbia, SC 29208, USA }
\email{haonanzhangmath@gmail.com}


\begin{abstract}
Polynomial threshold functions (PTFs) are an important low-complexity class of Boolean functions, with strong connections to learning theory and approximation theory.
Recent work on learning and testing PTFs has exploited structural and isoperimetric properties of the class, especially bounds on average sensitivity, one of the central themes in the study 
of PTFs since the Gotsman--Linial conjecture recently solved by OpenAI, \cite{GLsol}.

In this work we study PTFs through the lens of the Boolean surface area (or Talagrand boundary)
\[
\mathbf{BSA}[f]=\mathbb{E}|\nabla f|=\mathbb{E}\sqrt{s_{f}(x)},
\]
a natural measure of vertex-boundary complexity on the discrete cube.
Our main result is that every degree-$d$ PTF has polylogarithmic Boolean surface area:
\[
\mathbf{BSA}[f]\le C_d(\log(en))^{C_d}.
\]
The proof is based on the PTF Restriction Lemma of Kabanets, Kane, and Lu \cite{KKL2017} and proceeds through a tail bound for the pointwise sensitivity.
In particular, it controls all subcritical fractional moments of the sensitivity.
We also record a random block partition principle for Boolean surface area and an alternative recursive argument following Kane's work \cite{DK} on average sensitivity, which independently yields the second bound 
\[
\mathbf{BSA}[f]\le \exp(C(d)\sqrt{\log n}).
\]
This bound was obtained in \cite{FSVZ} without clarifying the magnitude of $C(d)$, which we do below.
For lower moments $\bE |\nabla f|^{\eta} \le C(d, \eta)$ there is no dependence on $n$ for $0<\eta <1$.
\end{abstract}

\thanks{F.C. is supported by the NSFC under grant 124B2019 and the Institute for Basic Science (IBS-R029-C4). The research of A.V. is supported  by 
Max Planck Institute for Mathematics in Bonn and NSF grant DMS-2154402. H.Z is supported by NSF DMS-2453408. }

\subjclass[2010]{42C10 (primary), 30L15, 46B07, 60G46}

\keywords{Polynomial Threshold Functions, Influence, Boolean surface area}

\maketitle

\section{Introduction}
\label{intro}

In this work we show a new constraint on the geometry of polynomial threshold functions by bounding their Boolean surface area.

\subsection*{Boolean surface area}
For any real-valued function $f$ on $\{-1,1\}^n$ and any $i\in [n]:=\{1,\dots, n\}$, its discrete derivative $D_i f$ is defined as
\begin{equation}
    D_i f(x)=\frac{f(x)-f(x^{\oplus i})}{2}, \qquad x=(x_1,\dots, x_n),
\end{equation}
where $x^{\oplus i}=(x_1,\dots, -x_i,\dots,x_n)$.
Boolean functions $f:\{-1,1\}^n\to \{-1,1\}$ play an essential role in theoretical computer science and other related areas, and one primary goal in this direction is to understand the structure of Boolean functions with small complexity.
In this paper, we do this with the so-called \emph{Boolean surface area}:
\[
\bsa[f]:= \bE|\nabla f|, \qquad\text{where}\qquad |\nabla f|=\sqrt{ \sum_{i=1}^n |D_i f|^2}\,.
\]
Here and in what follows, $\bE f=\bE f(x)$ with $x\sim \{-1,1\}^n$ being the uniform distribution.

A notion closely related to $\bsa$ is the \emph{total influence}.
Recall that the influence of the $i$-th coordinate of $f:\{-1,1\}^n\to \{-1,1\}$ is
\[
\Inf_i[f]:=\Pr_{x\sim \{-1,1\}^n}[f(x)\neq f(x^{\oplus i})] =\bE |D_if|^2,
\]
and the total influence is
\[
\Inf[f] =\sum_{i=1}^n \Inf_i[f]
=\bE |\nabla f|^2\,.
\]
Naively, one has the estimate
\begin{equation}\label{holder}
   \bsa[f] \le \sqrt{\Inf[f]}.
\end{equation}

While $\Inf[f]$ measures the size of the edge boundary of the set $A:=\{x:f(x)=-1\}$, $\bsa[f]$ gives some information about the \emph{vertex boundary} of $A$.
It is the central quantity in the works of Talagrand \cite{T} (building on Margulis \cite{M}) and Eldan--Gross \cite{EG}; see also \cite{EKLM}.
We call it the Boolean surface area in analogy to the Gaussian surface area (see Appendix~E of \cite{KODS} for an elaboration of this connection), but it is also called the \emph{Talagrand boundary} in \cite{EKLM}.
It is also the $\frac12$-moment of the sensitivity,
\begin{equation}
    \sens{f}(x):=\# \{i:f(x)\neq f(x^{\oplus i})\}=|\nabla f|^2(x).
\end{equation}
We remark that the total influence coincides with the \emph{average sensitivity}:
\[
\as[f]=\bE\, \sens{f}(x).
\]

\subsection*{The Boolean surface area of PTFs}
In this work we study $\bsa$ for Boolean functions $f$ computed by \emph{polynomial threshold functions} (PTFs) of small degree, namely
\[
f(x) := \sgn(p(x)),\quad x\in \{-1,1\}^n,
\]
where $p$ is a (multilinear) polynomial $\{-1,1\}^n\to \R$ of small degree.
In the following, we shall consider \emph{PTFs of degree $d$}, that is, $f=\sgn (p)$, and $p$ is a (multilinear) polynomial on $\{-1,1\}^n$ of degree at most $d$.
Here, $d$ is an integer that is small compared with the dimension $n$.
A special example is the Majority function
\begin{equation}
    \mathrm{MAJ}_n(x)=\sgn(x_1+\cdots +x_n),
\end{equation}
for which $\deg (p)=1$.

\bigskip

The well-known \emph{Gotsman--Linial conjecture} states that the extremal examples among degree-$d$ PTFs for total influence (average sensitivity) are symmetric polynomials that alternate signs around the middle levels of the discrete hypercube.
While this strong, structural formulation was proved false by Chapman \cite{C}, weaker versions of the Gotsman--Linial conjecture, about the maximum \emph{value} of total influence, remained open, with Kane's work \cite{DK} being the strongest step towards their proof, up to September 25, 2026 when OpenAI paper \cite{GLsol} appeared.
In particular, Kane proved that \cite{DK}
\begin{equation}\label{kane}
\as[f]\le n^{1/2}(\log n)^{O(d\log d)}\cdot 2^{O(d^2\log d)}\,.
\end{equation}
A popularly conjectured bound was $\sqrt{n}\cdot O(d)$, or, more weakly $\sqrt{n}\cdot O_d(1)$ (see \cite{ODopen}). Now after \cite{GLsol} we know that
\begin{equation}\label{openai-as}
\as[f]\le 8d\,n^{1/2}\,.
\end{equation}

\bigskip

The result of Kane however is the main inspiration of the present work, and we consider an analogous problem for $\bsa$. 

An interesting note is that the solution of Gotsman--Linial conjecture by AI in \cite{GLsol}
seems to not having any approach to estimating $\bsa$.

\smallskip

In the case of \emph{linear threshold functions} (LTFs), that is, $f=\sgn (p)$ and $p$ is linear, one might expect that because $\bsa(\mathrm{MAJ}_n)=\Theta(1)$, all LTFs must have constant $\bsa$.
However, this is not true; Klivans, O'Donnell, and Servedio proved in \cite{KODS} that the $\bsa$ of all LTFs are bounded by $\Theta(\sqrt{\log n})$, and this is optimal.
In particular, for $f = \sgn(\sum_i (\sqrt{i} -\sqrt{i+1})x_i$, one has $\bsa[f]=\Theta(\sqrt{\log n})$. 

\medskip

The second inspiration for this work was the following estimate in $L^1$ for small $t$ of the difference of $g$ and $P_t g$, where $P_t$ is the heat flow on Hamming cube:
\begin{equation}
\label{L1}
\bE |g- P_t g| \le \sqrt{2t} \,\bsa(g),
\end{equation}
which can be easily deduced from \cite{IVHV}. This is a simple but powerful estimate, which, for example, 
helped Ramon van Handel to greatly simplify the result of Eldan--Gross \cite{EG} solving a conjecture of Talagrand.
It also played an important part in \cite{IVHV} that solved Enflo's conjecture. 

\medskip

For general $d\ge 2$, it seems that prior to this work, the best off-the-shelf upper bound comes by combining Jensen's inequality and Kane's average sensitivity bound \eqref{kane}:
\[
\bsa[f]\le \sqrt{\as[f]} \le n^{1/4}(\log n)^{O(d\log d)}\cdot 2^{O(d^2\log d)}\,.
\]

One of the main result of this paper is the following.
\begin{theorem}
    \label{thm:main-est}
    For every fixed $d\ge 1$, there exists a constant $K_d>0$ depending only on $d$ such that every degree-$d$ polynomial threshold function $f:\{-1,1\}^n\to\{-1,1\}$ satisfies
\begin{equation}
\label{polylog-main}
\bsa[f] \le 32(\log(en))^{2K_d+1}.
\end{equation}
\end{theorem}

Our more precise result is a tail estimate in which the logarithm depends on the tail parameter, rather than on the dimension. We choose $K_d$ sufficiently large that the same constant can be used in both theorems.

\begin{theorem}\label{thm:dimension-free-tail}
Let $f:\{-1,1\}^n\to\{-1,1\}$ be a degree-$d$ PTF. For every integer $1\le m\le n$,
\begin{equation}\label{eq:dimension-free-tail}
\mathbb{P}[\sens{f}\ge m]\le\min\left\{1,\frac{C }{\sqrt{m}}(\log m)^{K_d}\right\}.
\end{equation}
\end{theorem}

\begin{corollary}\label{cor:small-moments}
For every $d\ge1$ and $0<\alpha<1/2$, there is a constant $C_{\alpha,d}$ such that every degree-$d$ PTF on every discrete cube satisfies $\mathbb{E}\sens{f}^{\alpha}\le C_{\alpha,d}$. Equivalently, $\mathbb{E}|\nabla f|^\eta\le C_{\eta,d}$ for $0<\eta<1$.
\end{corollary}

The key input in our proof is the PTF Restriction Lemma of Kabanets, Kane, and Lu \cite{KKL2017}.
For $0<r<1$, let $R_r$ denote the random restriction that leaves each coordinate free with probability $r$, independently, and otherwise fixes it to a uniformly random sign.
Roughly speaking, the restriction lemma says that once only an $r$-fraction of variables remain alive, a degree-$d$ PTF typically collapses to an almost constant function, and the failure probability is controlled at the optimal $\sqrt r$ scale.

In our proof, we shall use the restriction lemma to control local boundary geometry.
The bridge is a tail estimate for the pointwise sensitivity.
Suppose that $\sens{f}(x)\ge m$.
Then there are at least $m$ coordinates whose flip changes the value of $f$ at $x$.
If we now apply a random restriction that keeps each coordinate alive with probability $1/m$, then the restricted function still has sensitivity at least $1$ at the surviving point with positive probability.
Another ingredient is an elementary combinatorial observation: if a Boolean function on $\ell$ variables is $\delta$-close to a constant, then the set of points where its sensitivity is at least $1$ has measure at most $(\ell+1)\delta$.
Combining these ingredients yields a tail estimate for $\sens{f}(x)$, and then the Boolean surface area follows by summing the tails via the elementary identity
\[
\mathbb{E}\sqrt{X}=\sum_{m\ge 1}\bigl(\sqrt m-\sqrt{m-1}\bigr)\bP[X\ge m],
\]
applied to the nonnegative integer-valued random variable $X=\sens{f}(x)$.
In fact, the same method controls all subcritical fractional moments $\bE\,\sens{f}(x)^\alpha$ for $0<\alpha<1/2$.

\medskip 

Sections~\ref{splitting} and \ref{proof ingredients} also develop a random block partition framework and a Kane-style recursive argument which independently yield another estimate
\[
\bsa[f]\le \exp(Cd \sqrt{\log d}\sqrt{\log n}).
\]
This second approach seems useful in its own right, even though it does not recover the polylogarithmic bound above. 
On the first glance this result is always weaker than \eqref{polylog-main}. But there is a narrow window of regimes $(d, n)$ where the previous estimate is better than \eqref{polylog-main}.

The sharp dependence on $d$ and the optimal polylogarithmic exponent remain mysterious, however we will obtain certain estimates.
Even for $d=1$, the sharp order is $\Theta(\sqrt{\log n})$ \cite{KODS}, and for higher degrees the structure of extremizers (or approximate extremizers) is not understood.
We also remark that the Gaussian version of this story is essentially fully understood: in \cite{DKgauss} Kane proved that the Gaussian surface area of degree-$d$ PTFs is at most $d/\sqrt{2\pi}$, which is sharp, including the constant.

As a corollary of our main theorem, we derive a bound on the noise sensitivity of PTFs of degree $d$.
Recall that the noise sensitivity of a Boolean function $f$ with parameter $\delta\in (0,1/2)$ can be written as
\begin{equation}
\label{NSde}
    \textnormal{NS}_{\delta}[f]=\frac{1}{2}\left(1-\bE [fP_{\log\frac1{1-2\delta}} (f)]  \right),
\end{equation}
where $P_t=e^{t\Delta}$, $\Delta =-\sum_{j=1}^{n}D_j$ is the heat semigroup on the discrete hypercube.
In \cite[Corollary 1.3]{DK}, Kane obtained the bound
\begin{equation}
\label{viaNS1}
\textnormal{NS}_{\delta}[f]\le C(d)\,\sqrt{\delta}  \Big(\log \frac1{\delta}\Big)^{c \,d \log d}\,.
\end{equation}

We derive the following estimate using Lemma 2.1 of \cite{IVHV} and the fact that for Boolean $f$ 
\[
\left(1-\bE [fP_{t} (f)]  \right) =\bE| f- P_t f|
\]
\begin{corollary}\label{corollary}
    Let $f=\sgn(p)$ be a degree-$d$ polynomial threshold function on $\{-1,1\}^n$.
    Then for small $\delta$
    \begin{equation}
    \label{nsBSA}
\textnormal{NS}_{\delta}[f] \le \frac12\sqrt{2\log\frac{1}{1-2\delta}}\, \bsa[f] \le  3\sqrt{\delta} \,\bsa[f].
    \end{equation}
    when $\delta \in (0, 0.4267)$.
\end{corollary}

\medskip

A further improvement was obtained in \cite{DIRX}.
Namely  the sharp Boolean $L^1$ Poincar\'e inequality of Durcik, Ivanisvili, Roos, and Xie \cite[Theorem~1.3]{DIRX} to random restrictions gives a bound in terms of the Boolean surface area of the particular function.

\begin{corollary}\label{corollary}
Every Boolean function $f:\{-1,1\}^n\to\{-1,1\}$ satisfies
\begin{equation}\label{nsBSA}
\textnormal{NS}_\delta[f]\le\sqrt{\frac\delta2}\,\bsa[f],
\qquad 0\le\delta\le\frac12.
\end{equation}
The constant $1/\sqrt2$ in front of $\sqrt\delta\,\bsa[f]$ is optimal uniformly over this interval, even for degree-one PTFs. In particular, for a degree-$d$ PTF,
$$
\textnormal{NS}_\delta[f]\le 16\sqrt{2\delta}\,(\log(en))^{2K_d+1}.
$$
\end{corollary}

We give the proof and the sharpness example in Section~\ref{main-proof}. For comparison, combining~\eqref{GL8d} with~\cite{YP} one gets the following estimate of the noise sensitivity of degree $d$ PTF's:
\begin{equation}
\label{nsGL}
\textnormal{NS}_{\delta}[f] \le 8d \sqrt{\delta}\,.
\end{equation}
It can be worse than \eqref{nsBSA} for PTF's with small $\bsa$.

\medskip

A recent paper by OpenAI  \cite{GLsol} proved Gotsman--Linial conjecture by showing the estimate
\begin{equation}
\label{GL8d}
\as[f] \le 8d \sqrt{n}\,.
\end{equation}
Combining this with \cite{YP} one gets the following estimate of the noise sensitivity of degree $d$ PTF's:
\begin{equation}
\label{nsGL}
\textnormal{NS}_{\delta}[f] \le 8d \sqrt{\delta}\,.
\end{equation}

It can be worse than \eqref{nsBSA} for PTF's with small $\bsa$.

\medskip

 Interestingly degree $d$ PTF $f$ must have considerable  $\bsa[f]$ if its $\as[f]$ is close to being maximal in the sense of \eqref{GL8d}.

In fact, suppose  $\as[f] =\Inf[f] \ge 6 d\sqrt{n}$ then $\bsa[f] \ge 6d$, which immediately follows from
\[
6d \sqrt{n} \le \as[f] \le \bE[\sqrt{s_f}\sqrt{s_f}] \le \sqrt{n} \,\bsa[f]\,.
\]

\begin{rem}
\label{rem-strange}
 Clearly this naive reasoning shows that if one establishes a $polylog(n)$ estimate from above for $\bsa[f]$ (which is done below), then automatically one gets
 $\sqrt{n}$ times $polylog(n)$ estimate from above for $\as[f]$, which is Kane's celebrated result.  So, the estimate of $\bsa$ allows to lift the estimate to $\as$.  \end{rem}
 
 \begin{rem}
 Of course after \cite{GLsol} this lifting to Kane's result from our result (proved below) actually lost its glamor. But at the end of the paper we have some applications where the use of $\bsa$ gives advantages.
 \end{rem}

 \bigskip

In Section~\ref{main-proof} we prove Theorem~\ref{thm:main-est} and Corollary~\ref{corollary}.
Section~\ref{splitting} develops the random block partition idea.
Section~\ref{proof ingredients} sketches a Kane-style proof of the weaker bound, with the deferred calculations collected in Appendix~\ref{app:weak-calcs}.
Appendix~\ref{boundary} discusses the geometry of the boundary when $\bsa$ is small.

\section{Main result and proofs}
\label{main-proof}

We now prove Theorem~\ref{thm:main-est}.
For a parameter $0 < \delta < 1$, we say that a Boolean function $f$ is $\delta$-close to a constant if, for some $a\in\{-1,1\}$, we have $f(x)=a$ for all but at most a $\delta$ fraction of Boolean inputs $x$.
We will use the following restriction lemma of Kabanets, Kane, and Lu.

\begin{lemma}[PTF Restriction Lemma {\cite[Lemma 1.5]{KKL2017}}]\label{lem:restriction}
Fix $d\ge 1$.
Then there exists a constant $K_d>0$ depending only on $d$ such that the following holds.
Let $f:\{-1,1\}^n\to\{-1,1\}$ be a degree-$d$ polynomial threshold function.
For every $0<r,\delta\le \frac{1}{16}$,
\begin{equation}
\bP_{\rho\sim R_r}\bigl[f_\rho \text{ is not }\delta\text{-close to a constant}\bigr]
\le(\sqrt r+\delta)\left(\log\left(\frac{1}{r}\right)\cdot\log\left(\frac{1}{\delta}\right)\right)^{K_d},
\end{equation}
where $R_r$ denotes the random restriction that leaves each coordinate free (unrestricted) with probability $r$, independently, and otherwise fixes it to a uniformly random sign.
\end{lemma}

\begin{lemma}\label{lem:close-constant-implies-few-sensitive}
Let $g:\{-1,1\}^\ell\to\{-1,1\}$ be $\delta$-close to a constant.
Then
\[
\bP[\sens{g}(y)\ge 1]\le (\ell+1)\delta.
\]
\end{lemma}

\begin{proof}
Choose $a\in\{-1,1\}$ such that $A:=\{y\in\{-1,1\}^\ell: g(y)\neq a\}$ has measure at most $\delta$.
If $\sens{g}(y)\ge 1$, then there exists $j\in[\ell]$ such that $g(y)\neq g(y^{\oplus j})$.
Hence at least one of the points $y$ and $y^{\oplus j}$ lies in $A$.
Therefore
\[
\{y:\sens{g}(y)\ge 1\}
\subseteq
A\cup \bigcup_{j=1}^\ell \left\{y:y^{\oplus j}\in A\right\}.
\]
The map $y\mapsto y^{\oplus j}$ is a bijection of $\{-1,1\}^\ell$, so each set $\{y:y^{\oplus j}\in A\}$ also has measure at most $\delta$.
A union bound now gives
\[
\bP[\sens{g}(y)\ge 1]\le (\ell+1)\delta.
\qedhere
\]
\end{proof}

\begin{prop}\label{prop:tail}
Let $f:\{-1,1\}^n\to\{-1,1\}$ be a degree-$d$ PTF, and let $K_d$ be as in Lemma~\ref{lem:restriction}.
Assume $n\ge 256$.
Then for every integer $m$ with $16\le m\le n$,
\[
\bP[\sens{f}(x)\ge m]
\le \frac{8(\log(en))^{2K_d}}{\sqrt{m}}.
\]
\end{prop}

\begin{proof}
Fix $m\in\{16,\dots,n\}$.
Sample $x\in\{-1,1\}^n$ uniformly at random.
Independently, sample a random set $J\subseteq[n]$ by keeping each coordinate with probability $\frac{1}{m}$, independently.
From $(x,J)$, form the restriction $\rho_{x,J}\in\{-1,1,*\}^n$ by
\[
(\rho_{x,J})_i=
\begin{cases}
* & \text{if } i\in J,\\
x_i & \text{if } i\notin J.
\end{cases}
\]
Then $\rho_{x,J}$ has exactly the distribution $R_{1/m}$, and conditional on $\rho_{x,J}$, the vector $x_J$ is uniform on the live coordinates of $\rho_{x,J}$.

Define the events
\[
E:=\{x:\sens{f}(x)\ge m\},
\qquad
B:=\left\{(x,J): \sens{f_{\rho_{x,J}}}(x_J)\ge 1\right\}.
\]
For a fixed $x\in E$, let $S(x):=\{i\in[n]: f(x)\neq f(x^{\oplus i})\}$.
Then $|S(x)|\ge m$.
If $J\cap S(x)\neq\emptyset$, choose $i\in J\cap S(x)$.
Since $i$ remains live under the restriction $\rho_{x,J}$, the point $x_J$ in the restricted cube corresponds to the original point $x$, and flipping the $i$-th live coordinate sends $x_J$ to the restricted point corresponding to $x^{\oplus i}$, while all non-live coordinates remain fixed to their values in $x$.
Hence $f_{\rho_{x,J}}(x_J)=f(x)\neq f(x^{\oplus i})=f_{\rho_{x,J}}\bigl((x_J)^{\oplus i}\bigr)$, so $\sens{f_{\rho_{x,J}}}(x_J)\ge 1$.
Therefore
\[
\bP_J[B\mid x]\ge\bP_J[J\cap S(x)\neq\emptyset]=1-\left(1-\frac{1}{m}\right)^{|S(x)|}\ge1-\left(1-\frac{1}{m}\right)^m\ge 1-\frac{1}{e}>\frac{1}{2}.
\]
Averaging over $x\in\{-1,1\}^n$, we obtain
\[
\bP[B]\ge \mathbb{E}_{x\in E}\bP_J(B\mid x) \ge \frac12\bP[E].
\]

Now set $\delta:=\frac{\sqrt m}{n}$.
Since $n\ge 256$ and $m\le n$, we have $0<\delta\le \frac{1}{16}$, so Lemma~\ref{lem:restriction} applies.
Condition on a restriction $\rho=\rho_{x,J}$, and let $\ell$ be the number of live variables of $\rho$.
If $f_\rho$ is $\delta$-close to a constant, then Lemma~\ref{lem:close-constant-implies-few-sensitive} gives
\[
\bP\bigl[B\mid \rho\bigr]=\bP_{y\in\{-1,1\}^\ell}[\sens{f_\rho}(y)\ge 1]\le(\ell+1)\delta.
\]
Hence, without conditioning,
\begin{equation}
\label{NotdeltaClose}
    \begin{split}
\bP[B]&\le\bP_{\rho\sim R_{1/m}}\bigl[f_\rho \text{ is not }\delta\text{-close to a constant}\bigr]+ \delta\cdot\E[\ell+1]\\
&=\bP_{\rho\sim R_{1/m}}\bigl[f_\rho \text{ is not }\delta\text{-close to a constant}\bigr]+ \delta\cdot\left(\frac{n}{m}+1\right),
    \end{split}
\end{equation}
since $\ell\sim \mathrm{Bin}(n,1/m)$ and $\E\ell=\frac{n}{m}$.
Applying Lemma~\ref{lem:restriction} with $r=\frac{1}{m}$ and $\delta=\frac{\sqrt m}{n}$, we get
\begin{equation*}
    \begin{split}
\bP[B]&\le\left(\frac{1}{\sqrt{m}}+\delta\right)\cdot\left(\log m\cdot\log\left(\frac{n}{\sqrt m}\right)\right)^{K_d}+\frac{1}{\sqrt m}+\frac{\sqrt m}{n}\\
&\le\frac{2}{\sqrt{m}}(\log(en))^{2K_d}+\frac{2}{\sqrt{m}}\le\frac{4}{\sqrt{m}}(\log(en))^{2K_d},
    \end{split}
\end{equation*}
as $(\log(en))^{2K_d}\ge 1$.
Since $\bP[E]\le 2\bP[B]$, we conclude that
\[
\bP[\sens{f}(x)\ge m]=\bP[E]\le
\frac{8}{\sqrt{m}}(\log(en))^{2K_d}.
\qedhere
\]
\end{proof}

\begin{proof}[Proof of Theorem~\ref{thm:main-est}]
As $|\nabla f|(x)=\sqrt{\sens{f}(x)}$, it is enough to bound $\E\sqrt{\sens{f}(x)}$.

If $n<256$, then $\sens{f}(x)\le n$ pointwise and we are done.
Thus we may assume $n\ge 256$.
For every nonnegative integer-valued random variable $X$,
\begin{equation}
\E\sqrt{X}=\sum_{m=1}^{\infty} \left(\sqrt m-\sqrt{m-1}\right)\bP[X\ge m].
\end{equation}
Applying this with $X=\sens{f}(x)$, and noting that $\sens{f}(x)\le n$, we get
\[
\E\sqrt{\sens{f}(x)}=\sum_{m=1}^{n}\bigl(\sqrt m-\sqrt{m-1}\bigr)\bP[\sens{f}(x)\ge m].
\]
The contribution of $m\le 15$ is at most $\sum_{m=1}^{15}\bigl(\sqrt m-\sqrt{m-1}\bigr)=\sqrt{15}<4$.
For $m\ge 16$, Proposition~\ref{prop:tail} and the bound
\[
\sqrt m-\sqrt{m-1}=\frac{1}{\sqrt m+\sqrt{m-1}}\le\frac{1}{\sqrt m}
\]
give
\[
\sum_{m=16}^{n}\bigl(\sqrt m-\sqrt{m-1}\bigr)\bP[\sens{f}(x)\ge m]
\le
8(\log(en))^{2K_d}\sum_{m=16}^{n}\frac1m.
\]
Since $\sum_{m=16}^{n}\frac1m\le \log n\le \log(en)$, we have
\[
\bsa[f]\le 4+8(\log(en))^{2K_d+1}\le32(\log(en))^{2K_d+1},
\]
because $\log(en)\ge 1$.
\end{proof}

\begin{rem}
\label{lowMoments1}
The tail bound above immediately controls all subcritical fractional moments of the sensitivity:
for every fixed $0<\alpha<\frac{1}{2}$,
\[
\E \sens{f}(x)^\alpha \le C_{\alpha,d}(\log(en))^{2K_d}.
\]
Thus the present method controls an entire family of fractional boundary functionals, not only $\bsa[f]$.
\end{rem}

\begin{proof}[Proof of Corollary~\ref{corollary}]
The proof is an immediate consequence of Theorem~\ref{thm:main-est} and the following estimate for Boolean $f$:
\begin{equation}
    2\textnormal{NS}_{\delta}[f]=\bE|f-P_{\log\frac1{1-2\delta}}(f)|\le \sqrt{2\log\frac1{1-2\delta}}\,\bsa[f] \le 2\cdot 3\,\sqrt{\delta}\,\bsa[f]
\end{equation}
for $\delta\in (0, 0.4267)$

Here, the first inequality follows from the key formula of $D_j P_t f$ obtained in Lemma 2.1 of \cite{IVHV}, and the second inequality holds for $\delta\in (0, 0.3983)$ by calculus.
For the equality, note that for Boolean $f$, $P_t f\in [-1,1]$, so,
\[
\bE|f-P_t (f)|
=\bE|1-f P_t (f)|
= 1- \bE[f P_t (f)]
=2\textnormal{NS}_{\delta}[f].
\]
We take here $t: e^{-t} = 1-2\delta$.
\end{proof}

Actually we can do a bit better in terms of numerical constant by using a better absolute constant in the estimate of noise sensitivity obtained in \cite{DIRX}.

With our normalization of the discrete derivatives, the sharp Boolean $L^1$ Poincar\'e inequality of \cite[Theorem~1.3]{DIRX} states that
\begin{equation}\label{eq:sharp-boolean-poincare}
\var(g)=\mathbb{E}|g-\mathbb{E}g|\le\bsa[g]
\end{equation}
for every $g:\{-1,1\}^\ell\to\{-1,1\}$.

Fix $0<\delta\le1/2$, set $r=2\delta$, and take $\rho\sim R_r$. It is easy to see that
\begin{equation}\label{eq:noise-restriction-variance}
2\textnormal{NS}_\delta[f]=\mathbb E_{\rho\sim R_{2\delta}}\var(f_\rho).
\end{equation}
We next estimate the average Boolean surface area under a random restriction. Equivalently, choose a uniform $x\in\{-1,1\}^n$ and an independent random set $J\subseteq[n]$ by including each coordinate with probability $r$, and fix the coordinates outside $J$ to their values in $x$. This is the restriction $\rho_{x,J}$ used above. Conditional on $\rho_{x,J}$, the input $x_J$ is uniform. Jensen's inequality, applied to $J$ with $x$ fixed, therefore gives
\begin{equation}\label{eq:restriction-bsa}
\mathbb{E}_{\rho\sim R_r}\bsa[f_\rho]=\mathbb{E}_x\mathbb{E}_J\sqrt{\sum_{i\in J}\mathbf{1}_{\{f(x)\ne f(x^{\oplus i})\}}}\le\mathbb{E}_x\sqrt{r\,\sens{f}(x)}=\sqrt r\,\bsa[f].
\end{equation}
Combining \eqref{eq:sharp-boolean-poincare}--\eqref{eq:restriction-bsa}, we obtain
$$
2\textnormal{NS}_\delta[f]=\mathbb{E}_{\rho\sim R_{2\delta}}\var(f_\rho)\le\mathbb E_{\rho\sim R_{2\delta}}\bsa[f_\rho]
\le\sqrt{2\delta}\,\bsa[f],
$$
which proves \eqref{nsBSA}.

For sharpness, take $f(x)=x_1$. Then $\bsa[f]=1$ and $\textnormal{NS}_\delta[f]=\delta$, so equality holds in \eqref{nsBSA} at $\delta=1/2$. Moreover,
$$
\lim_{\delta\uparrow1/2}\frac{\textnormal{NS}_\delta[f]}{\sqrt\delta\,\bsa[f]}=\frac1{\sqrt2}.
$$
Thus the constant cannot be decreased even if the endpoint $1/2$ is excluded.

\section{An improvement of the previous estimate}
\label{impr}

Let us keep the previous notations, in particular, sets $E$ and $B$ have the same meaning:
\[
E:=\{x:\sens{f}(x)\ge m\},
\qquad
B:=\left\{(x,J): \sens{f_{\rho_{x,J}}}(x_J)\ge 1\right\}.
\]

However instead of \eqref{NotdeltaClose} let us write a similar estimate but not in the sense of $f_{\rho, J}$ being not $\delta$-close to constants but in the sense of a bigger set, namely of the set of $f_{\rho, J}$ not being $\delta$-concentrated. Namely, we already saw that $\bP(E) \le 2 \bP(B)$ and now we use

\begin{equation} 
\label{NotdeltaConc}
    \begin{split}
\bP[B]&\le\bP_{\rho\sim R_{1/m}}\bigl[f_\rho \text{ is not }\delta\text{-concentrated }\bigr] +\\
& \bP\bigl[\sens{f_{\rho_{x,J}}}(x_J)\ge 1 | f_\rho \text{ is  }\delta\text{-concentrated }\bigr] \cdot \bP_{\rho\sim R_{1/m}} \bigl[ f_\rho \text{ is  }\delta\text{-concentrated }\bigr]\,.
    \end{split}
\end{equation}

Now let us estimate $\bP \bigl[\sens{f_{\rho_{x,J}}}(x_J)\ge 1 | f_\rho \text{ is  }\delta\text{-concentrated }\bigr] $.

\begin{lemma}
\label{sf1}
Let $p$, $\deg p \le d$,  be $\delta$-concentrated. That is let
\begin{equation}
\label{deltaC}
\bE|p-\bE p|^2 \le \Big(\frac1{L\log\frac1{\delta}}\Big)^{d} \bE|p|^2\,. 
\end{equation}
Then  for $f=\sgn p$
\begin{equation}
\label{sf1-eq}
\bP\bigl[ s_{f}\ge 1\bigr] \le  2\,\delta\,.
\end{equation}
\end{lemma}

\begin{proof}
In fact, this is proved in \cite{KKL2017}, but we give the proof for the convenience. As always denote $\mu=\bE p$. Temporarily call 
 $\eps:= \Big(\frac1{L\log\frac1{\delta}}\Big)^{d}$. Then
\[
\bE|p|^2 -\mu^2 \le \eps \bE |p|^2\Rightarrow \mu^2 \ge (1-\eps) \bE|p|^2 \ge \Big(\frac1\eps -1) \bE|p-\mu|^2\,.
\]
We use here \eqref{deltaC} twice.
Thus we established that (given that $\delta$ is sufficiently small)
\begin{equation}
\label{mu}
\mu^2 \ge \frac12\Big(L\log\frac1{\delta}\Big)^{d} \bE|p-\mu|^2 = \frac12\frac1\eps \bE|p-\mu|^2\,.
\end{equation}

\medskip

On the other hand, $s_{f}(x)\ge 1$ means that $p(x)\to p(x^{\oplus i})$ changes sign for some $i$. In particular this means that
\[
\exists i\in [n],\quad |\partial_i p(x)| >\frac12|p(x)|\,.
\]
But we need to compare $|\partial_i p(x)|$ not with $|p(x)$ but with $|\mu|$.  We can have 

\medskip

case 1: $|p(x)| \ge |\mu|/2$, then $|\nabla p(x)| \ge |\mu|/4$;

\medskip

case 2:  $|p(x)| < |\mu|/2$, then $|p(x)-\mu| \ge |\mu|/2$.

In Event 1 we replaced $|\partial_i p|$ by $|\nabla p| =\big(\sum_{i=1}^n |\partial_i p|^2\big)^{1/2}$ because we do not know 
for what $i$ the inequality $|\partial_i p(x)| >\frac12|p(x)|$ happens.

\medskip

It is especially easy to estimate the probability of case 2. In this case we use that $|\mu|$ is very large, see \eqref{mu}, to write that

\begin{equation}
\label{set}
|p(x)-\mu|^2 \ge \frac18 \frac1\eps \bE|p-\mu|^2\,.
\end{equation}

The probability of such an even is given by concentration estimate for polynomial of degree $d$, it is 
\[
\le e^{-\frac{d}{2e(8\eps)^{1/d}} }\le  \delta^{Ld/(16 e)}<\delta
\]
by the choice of $L$.

It is a corollary hypercontractivity, and let us indicate the reasoning here.

Let $t>2$ and hypercontractivity for polynomials of degree $d$ gives

\[
\big(\bE|p-\mu|^t\big)^{1/t} \le (t-1)^{d/2}  \big(\bE|p-\mu|^2 \big)^{1/2} \,.
\]
On the other hand, the set in \eqref{set} is where
\[
|p-\mu|^t \ge \Big(\frac1{8\eps}\Big)^{t/2} \big(\bE|p-\mu|^2\big)^{t/2}.
\]
Combine the last two displays to get
\[
\frac{1}{\sqrt{\eps}} \big(\bE|p-\mu|^2 \big)^{1/2} \bP\bigl[\text{Event}\, 2\bigr] ^{1/t}  \le 2\sqrt{2}(t-1)^{d/2}  \big(\bE|p-\mu|^2 \big)^{1/2},
\]
or cancelling $(\bE|p-\mu|^2)^{1/2}$  from both parts
\[
\bP\bigl[\text{Event}\, 2\bigr]  \le 8^{t/2}(t-1)^{\frac{td}{2}} \, \eps^{\frac{t}2} = (8\varepsilon)^{\frac{t}{2}}(t-1)^{\frac{td}{2}} \,.
\]
Now one can take $t$ such that $e(t-1) = (8\varepsilon)^{-1/d}$  to get 
\[ 
\bP\bigl[\text{Event}\, 2\bigr]  \le  e^{-\frac{d}{2e(8\eps)^{1/d}} }
\le  \delta^{Ld/(16 e)}<\delta\,.
\]

To estimate $\bP\bigl[\text{Event}\, 1\bigr] $ we can follow the same steps, but we will need the comparison between $\bE|p-\mu|^2 $ and $\bE |\partial_i p|^2$.
We will need this to cancel $(\bE|p-\mu|^2)^{1/2}$ as we did this above.

\medskip

But it is standard that
\[
\deg p\le d\Rightarrow  \bE |\nabla p|^2 \le d \, \bE|p-\mu|^2 \,.
\]

This only adds the factor d to our estimate.

For each $i \in [n]$, we have
\[
\bP\left[|D_i p(x)| \ge \frac{|\mu|}{4}\right] = \bP\left[|D_i p(x)|^t \ge \left(\frac{|\mu|}{4}\right)^t\right]
\le \left(\frac{4}{|\mu|}\right)^t \E[|D_i p|^t]
\]
Using hypercontractivity and $t>2$, we can bound the right hand side by
\[
\left(\frac{4}{|\mu|}\right)^t (t-1)^{\frac{dt}{2}} \E[|D_i p|^2]^{t/2}
\].
Notice that
\[
\bP\bigl[\text{Event}\, 1\bigr] \le \sum_{i=1}^n \bP\left[|D_i p(x)| \ge \frac{|\mu|}{4}\right] \le \left(\frac{4}{|\mu|}\right)^t (t-1)^{\frac{dt}{2}}
\sum_{i=1}^n \E[|D_i p|^2]^{t/2}
\]
Since the $\ell^t$ norm is bounded above by the $\ell^2$ norm, we can bound the above by
\[
\left(\frac{4}{|\mu|}\right)^t (t-1)^{\frac{dt}{2}}
\left(\sum_{i=1}^n \E[|D_i p|^2] \right)^{\frac{t}{2}}
= \left(\frac{4}{|\mu|}\right)^t (t-1)^{\frac{dt}{2}} (\E[|\nabla p|^2])^{\frac{t}2}.
\]
Since $p$ has degree at most $d$,
\[
\bE |\nabla p|^2 \le d \, \bE|p-\mu|^2.
\]
Therefore, we have
\[
\bP\bigl[\text{Event}\, 1\bigr] \le \left(\frac{4}{|\mu|}\right)^t (t-1)^{\frac{dt}{2}} \left( d \, \bE|p-\mu|^2 \right)^{\frac{t}{2}}.
\]
Apply \eqref{mu}, we obtain
\[
\bP\bigl[\text{Event}\, 1\bigr] \le
4^t (t-1)^{\frac{dt}{2}} (2d\varepsilon)^{\frac{t}{2}} = (32d \varepsilon)^{\frac{t}{2}} (t-1)^{\frac{dt}{2}}.
\]
Take $t$ such that $e(t-1) = (32d\varepsilon)^{-1/d}$, which is really similar as in Event 2, we derive
\[
\bP\bigl[\text{Event}\, 1\bigr] \le
e^{-\frac{dt}{2}} \le
e^{-\frac{d}{2e(32 d\eps)^{1/d}} }
\le \delta^{Ld/(64 e)}
<\delta.
\]

\end{proof}

\medskip

Now we can go back to \eqref{NotdeltaConc}, and use the estimate of Lemma 5.1 of \cite{KKL2017} to write
\begin{equation}
\label{NotdeltaConc1}
    \begin{split}
\bP[B]&\le\bP_{\rho\sim R_{1/m}}\bigl[f_\rho \text{ is not }\delta\text{-concentrated }\bigr] +
\\
& \bE_{\rho\sim R_{1/m}}\Big( \bP\bigl[\sens{f_{\rho_{x,J}}}(x_J)\ge 1 | f_\rho \text{ is  }\delta\text{-concentrated }\bigr] \bP\bigl[ f_\rho \text{ is  }\delta\text{-concentrated }\bigr] |\rho\Big)\le
\\
& \Big(\frac1{\sqrt m} +\delta\Big) \cdot\left(\log m\cdot\log\left(\frac1{\delta}\right)\right)^{K_d} +  2\,\delta \,.
    \end{split}
    \end{equation}
    
    Now choose $\delta =\frac1{\sqrt m}$ to get finally 
    \[
    \bP\bigl[ s_{f} \ge m\bigr] \le 2 \bP\bigl[ B\bigr] \le \frac{C }{\sqrt{m}}(\log m)^{K_d}\,.
    \]

This does not depend on $\log n$.    So, the moments $\eta<1/2$ do not depend on $n$.

\medskip

Here is the improvement of the result \eqref{lowMoments1}.

\begin{rem}
\label{lowMoments2}
The tail bound above immediately controls all subcritical fractional moments of the sensitivity:
for every fixed $0<\alpha<\frac{1}{2}$,
\[
\E \sens{f}(x)^\alpha \le C_{\alpha,d}\,.
\]
Thus the present method controls an entire family of fractional boundary functionals, not only $\bsa[f]$.
\end{rem}

\section{The random block partition idea}
\label{splitting}

We are going to give another estimate for $\bsa(f)$. For certain regimes of $(d, n)$ it is better than the estimate we obtained above.
This will be explained in Section \ref{comparison}

\begin{theorem}
\label{thm:weak-est1}
Let $f=\sgn(p)$ be a degree-$d$ polynomial threshold function on $\{-1,1\}^n$.
Then
\[
\bsa[f]\le e^{Cd \sqrt{\textcolor{red}{\log(ed)}}\sqrt{\log n}}.
\]
\end{theorem}

\subsection{The special case: equal partition}

Let $\{y_j\}_{j=1}^n$ be a sequence of $k$ zeros and $n-k$ ones. Let $n=bm$ with $b,m\ge 1$ being integers. 
We wish to compare 
\begin{equation}\label{eq:AB}
A=\sqrt{\sum_{j=1}^{n} y_j}=\sqrt{n-k}\quad \text{and}\quad B=\frac{1}{\sqrt{b}}\bE_{\G}\sum_{\ell=1}^b \sqrt{\sum_{j\in \G_\ell} y_j}\,,
\end{equation}
where $\bE_{\G}$ is with respect to all partitions $\G=\{\G_1,\dots, \G_b\}$ of $[n]$ each having exactly $m$ elements. For any fixed splitting, applying the elementary estimate 
$$
\frac{1}{\sqrt{b}}\sum_{\ell=1}^{b}x_\ell\le \sqrt{\sum_{\ell=1}^{b}x_\ell^2}\le \sum_{\ell=1}^{b}x_\ell
$$
to $x_\ell=\sqrt{\sum_{j\in G_\ell}y_j}$ yields 
$$
B\le A\le \sqrt{b}B.
$$
It turns out that by taking the average $\bE_{\G}$ over all splittings, we can improve the upper bound to match the lower bound up to a small error.




\begin{prop}\label{prop:AvsB}
Under the above notation \eqref{eq:AB}, we have
\begin{equation}\label{eq:AvsB}
B\le A\le B+b
\end{equation}
for all $0\le k\le n$ and $n=mb$.
\end{prop}

To prove this, we first rewrite $B$ in a simplified form. 
Recall that the hypergeometric distribution $\hg(n, k, m)$: given a set of $n$ objects having $n-k$ successes, we choose $m$ objects at random and 
$X \sim \hg(n, n-k, m)$ is the distribution of successes in our chosen set of $m$ elements. 
For $X\sim \hg(n,n-k, m)$ the probability of $X=s$ is
\begin{equation}
\label{Hg}
\bP[X=s] =\frac{{n-k\choose s}{k\choose m-s}}{{n \choose m}}\,.
\end{equation}

To compute $B$, note that the number of splittings is
$$
\frac{\binom{n}{m}\binom{n-m}{m}\cdots \binom{2m}{m}}{b!},
$$
and each group $G\subset [n]$ of cardinality $|G|=m$ appears in exactly
$$
\frac{\binom{n-m}{m}\binom{n-2m}{m}\cdots \binom{2m}{m}}{(b-1)!}
=\frac{b}{\binom{n}{m}}\cdot\frac{\binom{n}{m}\binom{n-m}{m}\cdots \binom{2m}{m}}{b!}
$$
splittings. Thus, by symmetry,
\begin{equation}\label{eq:symmetry}
B=\frac{1}{\sqrt{b}}\bE_{\G}\sum_{\ell=1}^b \sqrt{\sum_{j\in G_\ell} y_j}
=\frac{\sqrt{b}}{\binom{n}{m}}\sum_{G\subset [n]: |G|=m}\sqrt{\sum_{j\in G}y_j}.
\end{equation}
For each $G\subset [n]$ containing $s$ zeros and $m-s$ ones, we have 
$$
\sqrt{\sum_{j\in G}y_j}=\sqrt{m-s}
$$ 
and the number of such $G$ is 
$$
\binom{k}{s}\binom{n-k}{m-s}.
$$
Here, the range of $s$ is
$$
0\le s \le \min \{m, k\}.
$$
So \eqref{eq:symmetry} and \eqref{Hg}  give
\begin{equation*}
B=\frac{\sqrt{b}}{\binom{n}{m}}\sum_{s=0}^{\min \{m,k\}}\binom{k}{s}\binom{n-k}{m-s}\sqrt{m-s}
=\sqrt{b}\bE_{X\sim \hg(n,k,m)}[\sqrt{m-X}],
\end{equation*}
or equivalently, 
\begin{equation}
\label{eq:B}
B=\sqrt{b}\bE_{X\sim \hg(n,n-k,m)}[\sqrt{X}].
\end{equation}

This form is much easier to work with, and we need the following lemma. 

\begin{lemma}\label{lem:jensen}
    Let $X$ be a nonzero random variable taking values in $[0,\infty)$. Then 
    \begin{equation}
       \sqrt{\bE X}-  \frac{1}{2}(\bE X)^{-3/2}\var(X) \le \bE\sqrt{X}\le  \sqrt{\bE X}.
    \end{equation} 
    Consequently, for constants $a,b$ such that $aX+b$ is nonzero taking values in $[0,\infty)$, one has 
    \begin{equation}
       \sqrt{a\bE X+b}-\frac{a^2}{2}(a\bE X+b)^{-3/2}\var(X) \le \bE\sqrt{aX+b}\le\sqrt{a\bE X+b}.
    \end{equation}
\end{lemma}

\begin{proof}
   The second statement follows from the first one by rescaling. To prove the first statement, note that the right-hand side estimate is simply the Jensen inequality for $\sqrt{x}$. For the left-hand side, we have for any $x_0>0$
   \begin{align*}
      \sqrt{x}-\sqrt{x_0}-\frac{x-x_0}{2\sqrt{x_0}}
       =&\frac{x-x_0}{\sqrt{x}+\sqrt{x_0}}-\frac{x-x_0}{2\sqrt{x_0}}\\
       =&\frac{(x-x_0)(\sqrt{x_0}-\sqrt{x})}{2\sqrt{x_0}(\sqrt{x}+\sqrt{x_0})}\\
       =&-\frac{(x-x_0)^2}{2\sqrt{x_0}(\sqrt{x}+\sqrt{x_0})^2} \\
       \ge& -\frac{1}{2}x_0^{-3/2}(x-x_0)^2.
   \end{align*}
   Taking the expectation $\bE_{x\sim X}$ on both sides with $x_0=\bE X$ finishes the proof.  
\end{proof}

Now we are ready to prove Proposition \ref{prop:AvsB}.

\begin{proof}[Proof of Proposition \ref{prop:AvsB}]
The left-hand side estimate $B\le A$ is trivial, as we remarked earlier. We focus on the right-hand side estimate 
\begin{equation}\label{eq:upper}
\sqrt{n-k}\le \sqrt{b}\bE[\sqrt{X}]+b
\end{equation}
recalling \eqref{eq:B}, where $X=\hg(n,n-k,m)$ be the hypergeometric distribution. This trivially holds when $n-k\le b^2$. 

Now let us assume  $b^2< n-k$. To prove \eqref{eq:upper} in this case, recall that 
\begin{equation}
    \bE X=\frac{m(n-k)}{n}=\frac{n-k}{b},\qquad \textnormal{and}\qquad \var(X)=\frac{mk(n-k)(n-m)}{n^2(n-1)}.
\end{equation}
According to Lemma \ref{lem:jensen}, we have
\begin{equation}
\sqrt{b}\sqrt{\bE X}- \sqrt{b} \bE\sqrt{X}
\le \frac{\sqrt{b}}{2}(\bE X)^{-3/2}\var(X)\,,
\end{equation}
which is nothing but 
\begin{equation}
    \sqrt{n-k}-\sqrt{b}\bE\sqrt{X}
    \le \frac{\sqrt{b}}{2}\left(\frac{n-k}{b}\right)^{-3/2}\frac{mk(n-k)(n-m)}{n^2(n-1)}\, .
\end{equation}
The right-hand side simplifies as (recalling $n=bm$)
\begin{equation}
 \!\!\!\!  \!\!\!\!  \frac{\sqrt{b}}{2}\left(\frac{n-k}{b}\right)^{-3/2}\frac{mk(n-k)(n-m)}{n^2(n-1)}
    =\frac{b}{2\sqrt{n-k}}\frac{k(n-m)}{n(n-1)}
    \le \frac{b}{2\sqrt{n-k}}.
\end{equation}
In case $b^2< n-k$, this is bounded by $1/2$, so that 
\begin{equation}
\sqrt{n-k}- \sqrt{b} \bE\sqrt{X}\le \frac{1}{2}\le b.
\end{equation}

Therefore, we always have \eqref{eq:upper}, and thus finish the proof. 
\end{proof}

\subsection{The general case: arbitrary partition}

In this subsection, we prove similar bounds when \([n]\) is split into \(b\) blocks of prescribed,
not necessarily equal, sizes. Let \(\{y_j\}_{j=1}^n\) be a sequence of \(k\) zeros and \(n-k\)
ones, and let
\[
m_1,\dots,m_b\ge 1,
\qquad
m_1+\cdots+m_b=n.
\]
We wish to compare
\begin{equation}\label{eq:AB-arbitrary-sizes}
A=\sqrt{\sum_{j=1}^n y_j}=\sqrt{n-k}
\quad\text{and}\quad
B=
B(m_1,\dots, m_b)=\frac{1}{\sqrt b}\,\bE_{\Pi}\sum_{\ell=1}^b \sqrt{\sum_{j\in \Pi_\ell} y_j},
\end{equation}
where \(\bE_{\Pi}\) is with respect to all ordered splittings
\(\Pi=(\Pi_1,\dots,\Pi_b)\) of \([n]\) such that \(|\Pi_\ell|=m_\ell\) for every
\(1\le \ell\le b\).

For each \(1\le \ell\le b\), let
\[
X_\ell\sim \mathrm{Hg}(n,n-k,m_\ell).
\]
For fixed \(\ell\), every set \(G\subseteq [n]\) of cardinality \(m_\ell\) appears equally
often as \(\Pi_\ell\). Hence
\[
\bE_{\Pi}\sqrt{\sum_{j\in \Pi_\ell} y_j}
=
\binom{n}{m_\ell}^{-1}
\sum_{G\subseteq [n]:\,|G|=m_\ell}
\sqrt{\sum_{j\in G} y_j}
=
\bE\sqrt{X_\ell}.
\]
Therefore
\begin{equation}\label{eq:B-arbitrary-sizes-hg}
B=B(m_1,\dots, m_b)=\frac{1}{\sqrt b}\sum_{\ell=1}^b \bE\sqrt{X_\ell}.
\end{equation}

\begin{prop}\label{prop:arbitrary-sizes}
Under the above notation, we have
\[
B\le A.
\]
Moreover, if \(k<n\), then
\begin{equation}\label{eq:AB-arbitrary-sizes-bound}
A-B
\le
\sqrt{n-k}\left(1-\frac{1}{\sqrt{bn}}\sum_{\ell=1}^b \sqrt{m_\ell}\right)
+
\frac{k}{2(n-1)\sqrt{bn(n-k)}}
\sum_{\ell=1}^b \frac{n-m_\ell}{\sqrt{m_\ell}}.
\end{equation}
\end{prop}

\begin{proof}
The lower bound \(B\le A\) is immediate, as before. Indeed, for every fixed splitting \(\Pi\),
\[
\frac{1}{\sqrt b}\sum_{\ell=1}^b \sqrt{\sum_{j\in \Pi_\ell} y_j}
\le
\sqrt{\sum_{\ell=1}^b \sum_{j\in \Pi_\ell} y_j}
=
\sqrt{\sum_{j=1}^n y_j}
=
A,
\]
by Cauchy--Schwarz inequality and averaging over \(\Pi\) proves $B\le A.$

Now assume \(k<n\). For each $1\le \ell\le b$, we have
\[
\bE X_\ell=\frac{m_\ell(n-k)}{n},
\qquad
\operatorname{Var}(X_\ell)=
\frac{m_\ell k (n-k)(n-m_\ell)}{n^2(n-1)}.
\]
Applying Lemma \ref{lem:jensen} to each \(X_\ell\), we obtain
\begin{align*}
\bE\sqrt{X_\ell}
\ge
\sqrt{\bE X_\ell}
-\frac{1}{2}(\bE X_\ell)^{-3/2}\operatorname{Var}(X_\ell)
=
\sqrt{\frac{m_\ell(n-k)}{n}}
-
\frac{k(n-m_\ell)}{2(n-1)\sqrt{n\,m_\ell\,(n-k)}}.
\end{align*}
Summing over \(\ell\) and using \eqref{eq:B-arbitrary-sizes-hg}, we get
\[
B
\ge
\sqrt{\frac{n-k}{bn}}\sum_{\ell=1}^b \sqrt{m_\ell}
-
\frac{k}{2(n-1)\sqrt{bn(n-k)}}
\sum_{\ell=1}^b \frac{n-m_\ell}{\sqrt{m_\ell}}.
\]
Subtracting from \(A=\sqrt{n-k}\) yields \eqref{eq:AB-arbitrary-sizes-bound}.
\end{proof}

\begin{corollary}\label{cor:almost-equal-sizes}
Suppose
\[
n=bm+r,
\qquad
0\le r<b,
\]
and the block sizes satisfy
\[
m_\ell\in\{m,m+1\}
\qquad (1\le \ell\le b),
\]
with exactly \(r\) of them equal to \(m+1\). Then
\begin{equation}\label{eq:B-almost-equal-sizes}
B=
\frac{b-r}{\sqrt b}\,\bE\sqrt{Y}
+
\frac{r}{\sqrt b}\,\bE\sqrt{Z},
\end{equation}
where
\[
Y\sim \mathrm{Hg}(n,n-k,m),
\qquad
Z\sim \mathrm{Hg}(n,n-k,m+1),
\]
and
\begin{align}
A-B
&\le
\sqrt{n-k}
\left(
1-\frac{(b-r)\sqrt m+r\sqrt{m+1}}{\sqrt{bn}}
\right)
\notag\\
&\qquad\qquad
+
\frac{k}{2(n-1)\sqrt{bn(n-k)}}
\left(
(b-r)\frac{n-m}{\sqrt m}
+
r\frac{n-m-1}{\sqrt{m+1}}
\right).
\label{eq:AB-almost-equal-sizes}
\end{align}
In particular,
\begin{equation}\label{eq:AB-almost-equal-sizes-coarse}
B\le A\le B+b.
\end{equation}
\end{corollary}

\begin{proof}
The representation \eqref{eq:B-almost-equal-sizes} is immediate from
\eqref{eq:B-arbitrary-sizes-hg}, and \eqref{eq:AB-almost-equal-sizes} is just
\eqref{eq:AB-arbitrary-sizes-bound} specialized to the present choice of block sizes.

It remains to prove $A\le B+b$. Again, if \(n-k\le b^2\), then
\[
A=\sqrt{n-k}\le b\le B+b.
\]
Now assume \(n-k>b^2\). Then \(n>b^2\), and therefore \(m=\lfloor n/b\rfloor\ge b\).

It suffices to show the right-hand side of \eqref{eq:AB-almost-equal-sizes}, denoted by \(T_1+T_2\), is no larger than $b$, where
\[
T_1=
\sqrt{n-k}
\left(
1-\frac{(b-r)\sqrt m+r\sqrt{m+1}}{\sqrt{bn}}
\right)
\]
and
\[
T_2=
\frac{k}{2(n-1)\sqrt{bn(n-k)}}
\left(
(b-r)\frac{n-m}{\sqrt m}
+
r\frac{n-m-1}{\sqrt{m+1}}
\right).
\]

For \(T_1\), note that
\[
\frac{(b-r)\sqrt m+r\sqrt{m+1}}{b}
\]
is the average of \(\sqrt m\) and \(\sqrt{m+1}\) with weights \((b-r)/b\) and \(r/b\),
respectively. Applying Lemma \ref{lem:jensen} to the random variable $X$  such that $\Pr[X=m]=(b-r)/b$  and $\Pr[X=m+1]=r/b$  yields
\[
\sqrt{\frac{n}{b}}
-
\frac{(b-r)\sqrt m+r\sqrt{m+1}}{b}
\le
\frac{1}{2}\left(\frac{n}{b}\right)^{-3/2}\frac{r(b-r)}{b^2}.
\]
Multiplying by \(\sqrt{b(n-k)/n}\), we obtain
\[
T_1\le \frac{r(b-r)\sqrt{n-k}}{2n^2}.
\]
Since \(r(b-r)\le b^2/4\) and \(n\ge n-k>b^2\), it follows that
\[
T_1\le \frac{b^2\sqrt n}{8n^2}<\frac18.
\]

For \(T_2\), we have
\[
T_2
\le
\frac{k}{2(n-1)\sqrt{bn(n-k)}}\,b\frac{n-m}{\sqrt m}
=
\frac{bk(n-m)}{2(n-1)\sqrt{bnm(n-k)}}.
\]
Since \(n\ge bm\), we have \(\sqrt{bnm}\ge bm\), hence
\[
T_2
\le
\frac{k(n-m)}{2m(n-1)\sqrt{n-k}}=\frac{n-m}{m\sqrt{n-k}}\cdot \frac{k}{2(n-1)}.
\]
Recall that  $n-m=(b-1)m+r<bm, b<\sqrt{n-k}$ and $k\le n-1$, so
\[
T_2
\le 
\frac12.
\]

Therefore
\[
A-B\le T_1+T_2<\frac18+\frac12<1\le b,
\]
which proves \eqref{eq:AB-almost-equal-sizes-coarse}.
\end{proof}

\section{The second bound for $\bsa$}
\label{proof ingredients}
\subsection{Random splitting}
With the key estimate Corollary \ref{cor:almost-equal-sizes}  in hand, the arc of our proof is similar to Kane's work \cite{DK} about the average sensitivity of $f=\sgn (p)$. For this, we begin by splitting the coordinates into $b$ blocks $G_1,\dots, G_b$, each of which has at most $n/b+1$ elements:
$$
[n]=\cup_{\ell =1}^{m}G_\ell,\qquad |G_\ell|\le \lfloor n/b\rfloor+1.
$$
We shall use the following notation. When $x\in \{-1,1\}^n$ is divided into two parts $x=(y,z)$, we write $f_y(z)$ for $f(x)=f(y,z)$. This way, any function $f$ in $x$ restricts to a function $f_y$ in $z$. In particular, for each Bernoulli random variable $A\in \{-1,1\}^n$ and block $G_\ell$, we let $A^{\ell}$ be the coordinates of $A$ that do not lie in $G_\ell$. Then $f_{A^\ell}$ defines a function on coordinates in $G_\ell$. 

Kane's argument starts with the elementary identity for average sensitivity
\begin{equation}
\as[f]=\sum_{\ell}\bE_{A^\ell} \as[f_{A^\ell}],
\end{equation}
which fails for $\bsa$. It is for this reason we use the substitute 
\begin{equation}
\bsa[f]\le \frac{1}{\sqrt{b}} \bE_{\G}\sum_{\ell=1}^{b}\bE_{A^{\ell}}\bsa[f_{A^\ell}]+b
\end{equation}
obtained by taking expectation of \eqref{eq:AB-almost-equal-sizes-coarse}. 

\subsection{The function $\alpha$}

The following function $\alpha$ plays a crucial role in Kane's proof. For a nonzero polynomial $p$ on $\{-1,1\}^n$ and a vector $v\in \{-1,1\}^n$, we define 
$$
D_v p(x):=\langle v, \nabla p(x) \rangle =\sum_{j=1}^{n}v_j D_j p(x).
$$
We then define $\alpha(p)$ as 
\begin{equation}
\alpha(p):=\bE\min \left(1, \frac{|D_B p(A)|^2}{|p(A)|^2} \right),
\end{equation}
where $A$ and $B$ are i.i.d. Bernoulli random variables.
The quantity $\alpha(p)$ will serve as a key parameter in the induction. In Kane's work \cite{DK}, he also needs its Gaussian variant and the invariance principle. Here, we omit the details and refer to Kane's original paper for discussion. 

\subsection{The regular case}

As before, let $n$ be the dimension of the discrete hypercube $d$ be the degree. 

\begin{defn}
For any $a>0$, we define $\maxbsa(d,n,a)$ as the maximum Boolean surface area of a PTF $f=\sgn (p)$, where $\deg(p)\le d$ and $\alpha(p)\le a$.
\end{defn}

We will also need a variant of $\maxbsa$ for regular polynomials. Recall that a polynomial $p$ is $\tau$-regular for some $\tau>0$ if $\Inf_i[p]\le \tau \var[p]$ for all $i\in [n]$.

\begin{defn}
For any $a,\tau>0$, we define $\maxbsar(d,n,a,\tau)$ as the maximum Boolean surface area of a PTF $f=\sgn (p)$, where $\deg(p)\le d$, $\alpha(p)\le a$ and $p$ is $\tau$-regular. 
\end{defn}

We shall use notations $\maxas$ and $\maxasr$ for average sensitivities in a similar manner. 

Similar to average sensitivity \cite{DK}, we have the following proposition.

\begin{prop}\label{prop1}
Let $a,\tau>0$ and $b\le n$ be a positive integer.  Then 
\begin{equation}
\maxbsar(d,n,a,\tau)\le \sqrt{b}\,\bE_{\aleph}\maxbsa(d,\lfloor n/b\rfloor+1,\aleph)+b
\end{equation}
for some random variable $\aleph$ with $\bE \aleph=O(d^3 a b^{-1/2}+d^4\tau^{\frac{1}{8d}})$
\end{prop}

\begin{proof}
The proof is identical to that of \cite[Proposition 4.1]{DK}, except that one replaces
\begin{equation}
\as[f]=\sum_{\ell}\bE_{A^\ell} \as[f_{A^\ell}]
\end{equation}
with
\begin{equation}
\bsa[f]\le \frac{1}{\sqrt{b}} \bE_{\G}\sum_{\ell=1}^{b}\bE_{A^{\ell}}\bsa[f_{A^\ell}]+b.
\end{equation}
\end{proof}

\subsection{The general case: reduction to the regular polynomials}

Following \cite{DK}, we have the following reduction result. 

\begin{prop}\label{prop2}
Let $0<a\le 1$ and $0<\tau,\epsilon<1/4$ and $b\le n$ be a positive integer.  Then 
\begin{align*}
\maxbsa(d,n,a)&\le \tau^{-1/2}\left(d\log(1/\tau)\log (1/\epsilon) \right)^{O(d)}+3\sqrt{n\epsilon}\\
&\qquad +\bE_{\aleph}[\maxbsar(d,n,\aleph,\tau)]
\end{align*}
for some nonnegative random variable $\aleph$ with $\bE[\aleph]\le a$.
\end{prop}

\begin{proof}
The proof is similar to that of \cite[Proposition 4.4]{DK}, which relies on \cite[Proposition 2.11]{DK} about decision-tree decomposition: Any polynomial $p$ on $\{-1,1\}^n$ of degree $d$ can be written as a decision tree of depth at most 
$$
D=\tau^{-1}\left(d\log(1/\tau)\log (1/\epsilon) \right)^{O(d)}
$$
with variables at the internal nodes such that for a random leaf $\rho$, with probability $1-\epsilon$, the polynomial $p_\rho$ is either $\tau$-regular, or constant sign with probability at least $1-\epsilon$. Here, $p_\rho$ is the function corresponding to the leaf $\rho$.

In the case of average sensitivity, Kane proved
\begin{equation}
\maxas(d,n,a)\le D+3n\epsilon+\bE_{\aleph}[\maxasr(d,n,\aleph,\tau)]
\end{equation}
via the pointwise estimate 
\begin{equation}\label{decisiontree}
\sens{f}(x)\le D+s_{f_\rho}(x_{\textnormal{free}}).
\end{equation}
Here, $x_{\textnormal{free}}$ denotes the coordinates that are not fixed by the leaf $\rho$. Unlike average sensitivity that is linear in $\sens{f}$, the Boolean surface area $\bsa[f]$ is the expectation of the square root of $\sens{f}$. But we still have 
\begin{equation}
\sqrt{\sens{f}(x)}\le \sqrt{D}+\sqrt{s_{f_\rho}(x_{\textnormal{free}})}.
\end{equation}
from \eqref{decisiontree}. Taking the expectation gives
\begin{equation}
\bsa[f]\le \sqrt{D}+\bE_{\textnormal{leaves } \rho}\bsa[f_{\rho}]
\end{equation}
Now we further estimate $\bE_{\textnormal{leaves } \rho}\bsa[f_{\rho}]$ as was done for $\bE_{\textnormal{leaves } \rho}\as[f_{\rho}]$ in \cite{DK}. Recall that with probability $1-\epsilon$, $p_\rho$ is either $\tau$-regular or constant sign with probability $1-\epsilon$, thus dividing the leaves into three parts: (1) the exceptional set of probability at most $\epsilon$, (2) the leaves for which $p_\rho$ has constant sign with probability $1-\epsilon$, and (3) the leaves that are $\tau$-regular. 

  The contribution from part (1) is at most $\sqrt{n}\epsilon$ (compared with $n\epsilon$ for average sensitivity). The contribution from part (2)  is at most $2\sqrt{n\epsilon}$ (compared with $2n\epsilon$ for average sensitivity). The contribution from part (3) is controlled by 
$$\bE_{\aleph}[\maxbsar(d,n,\aleph,\tau)]$$
with $\bE[\alpha(p_\rho)]\le C \bE[\alpha(p)] \le C\,a$. All combined, we finish the proof. 
\end{proof}

\subsection{Putting everything together}
\label{putting together}

We start with a variant of Lemma~4.5 of \cite{DK}.

\begin{lemma}\label{lem:small alpha}
Let $f=\sgn(p)$ with $\deg(p)\le d$ and
\[
\alpha(p)\le (K\log n)^{-d},
\qquad K\gg 1.
\]
Then
\[
\bsa[f]\le \alpha(p).
\]
\end{lemma}

\begin{proof}
By the proof of [Kan14, Lemma 4.5], increasing \(K\) if necessary, we have $\as[f]\le\alpha(p)$. Since $\bsa[f]\le\as[f]$, the result follows.
\end{proof}

The preceding ingredients yield an alternative proof of a weaker bound.

\begin{theorem}\label{thm:weak-est}
Let $f=\sgn(p)$ be a degree-$d$ polynomial threshold function on $\{-1,1\}^n$.
Then
\[
\bsa[f]\le e^{C(d)\sqrt{\log n}},
\]
where $C(d)\le A\,d\sqrt{\log(ed)}$ for an absolute constant $A$.
For sufficiently large $d$, one may take $A=3$.
\end{theorem}

\begin{rem}
We use $\log(ed)$ only to avoid degeneracy at $d=1$.
In the large-degree argument below, we use $\log d$ and prove the stronger estimate $C(d)=3d\sqrt{\log d}$.
\end{rem}

\begin{proof}
It suffices to prove the sharper estimate for sufficiently large $d$.
We follow the argument of Section~4 of \cite{DK}, using Proposition~\ref{prop1}, Proposition~\ref{prop2}, and Lemma~\ref{lem:small alpha}.
Set
\[
F(n,a):=\maxbsa(d,n,a),
\qquad
A(n):=(K\log n)^{\textcolor{red}{-d}},
\qquad
b(n):=e^{\sqrt{\log n}},
\]
Define
\begin{equation}
\label{tau}
\tau(n):=\frac{1}{(K\log n)^{16d} b(n)^{4d + 4\sqrt2\,d\sqrt{\log d}}},
\end{equation}
as well as
\[
P(n):=\tau(n)^{-1/2}\Bigl(d\log \frac1{\tau(n)} \cdot \log n\Bigr)^{O(d)}.
\]
For $\log n\ge M_d^2$, with $M_d$ chosen below, applying Proposition~\ref{prop2} with $\epsilon=1/n$, and then Proposition~\ref{prop1}, we obtain
\[
F(n,a)
\le
P(n)+3+b(n)
+\sqrt{b(n)}\,\E_{\aleph}\!\bigl[F(\lfloor n/b(n)\rfloor+1,\aleph)\bigr],
\]
where
\[
\E \aleph \le C_d\!\left(a\,b(n)^{-1/2}+\tau(n)^{1/(8d)}\right).
\]
Since
\[
\tau(n)^{1/(8d)}
= (K\log n)^{-2}b(n)^{-1/2}
b(n)^{-\sqrt{(\log d)/2}},
\]
we have
\[
\sqrt{b(n)}\,\E \aleph
\le C_d\bigl(a+A(n)\bigr),
\]
if
\[
\sqrt{(\log d)/2}\sqrt{\log n}\ge(d-2)\log(K\log n).
\]
For sufficiently large $d$, this holds whenever
\begin{equation}
\label{sqrt}
\log n\ge M_d^2,\qquad M_d=\textcolor{red}{3}d\sqrt{\log d}.
\end{equation}
Indeed, at $\log n=M_d^2$ the leading coefficients on the left and on the right are $3/\sqrt2>2$, and the inequality persists as $n$ increases.

\bigskip

We claim that for such  $M=M_d$,
\[
F(n,a)\le a\,\Phi(n),
\qquad
\Phi(n):=e^{M\sqrt{\log n}},
\]
for all $0\le a\le1$.
The proof is by induction on $n$.
The initial step is verified in Appendix~\ref{app:weak-calcs}.
For the induction step, assume the claim holds in smaller dimension.
Then for every realization $u$ of $\aleph$,
\[
F(\lfloor n/b(n)\rfloor+1,u)\le u\,\Phi(\lfloor n/b(n)\rfloor+1),
\]
so
\[
\E_{\aleph}\!\bigl[F(\lfloor n/b(n)\rfloor+1,\aleph)\bigr]
\le
\E \aleph \cdot \Phi(\lfloor n/b(n)\rfloor+1).
\]
Substituting into the recurrence gives
\[
F(n,a)
\le
P(n)+3+b(n)
+
C_d\bigl(a+A(n)\bigr)\Phi(\lfloor n/b(n)\rfloor+1).
\]
(Here $C_d=O(d^4)$. See Proposition \ref{prop1}.)

If $a\le A(n)$, then Lemma~\ref{lem:small alpha} yields $F(n,a)\le a\le a\Phi(n)$.
Hence we may assume $a>A(n)$, in which case
\[
F(n,a)
\le
\underbrace{P(n)+3+b(n)}_{=:\textnormal{I}}
+
\underbrace{\textcolor{red}{2}C_d\,a\,\Phi(\lfloor n/b(n)\rfloor+1)}_{=:\textnormal{II}}.
\]
Appendix~\ref{app:weak-calcs} proves that, with $M=M_d$ as in \eqref{sqrt} and $d$ sufficiently large,
\[
\textnormal{II}\le \frac12 a\,\Phi(n),
\qquad
\textnormal{I}\le \frac12 a\,\Phi(n).
\]
Therefore $F(n,a)\le a\Phi(n)$, proving the claim.
Choosing $a=1$ yields the theorem.
\end{proof}

\section{Comparison of polylog bound and block splitting bound for $\bsa$}
\label{comparison}
 We wish to compare two quantities
 \begin{equation}
 \label{polyL}
 Q_{polylog} = c (1+\log n)^{Bd^2} = c e^{Bd^2 \log(1+\log n)}
 \end{equation}
 and
 \begin{equation}
 \label{polyL}
 Q_{block} = e^{A d\sqrt{\log d} \sqrt{\log n}}\,.
 \end{equation}
 
 We also wish to look at them only in the $(d, n)$ regime where both of them is smaller than $n^{1/4}$. So, the regime where
 \[
 d^2\ge \log n
 \]
 is excluded. Therefore we can write
 \[
 Q_{polylog} \ge Q_{pl}:= e^{2B d^2 \log d}\,.
 \]
 Now we get interested in the regimes (if exist), where $d^2\le \log n$ and 
 \begin{equation}
 \label{compa1}
 Q_{block} < Q_{pl} < n^{1/4}\,.
 \end{equation}
 An essential variable will be
 \[
 t:= \frac{\log n}{d^2\log d}\,.
 \]
 Notice that in terms of $t$
 \[
 Q_{block} =n^{\frac{A}{\sqrt{t}}}, \,\,  Q_{pl} = n^{\frac{2B}{t}}\,.
 \]
 
 So the inequalities \eqref{compa1} become
 \[
 8B< t=\frac{\log n}{d^2\log d}< \frac{4B^2}{A^2},
 \]
 which is reconcilable  if
 $$
 B>2A^2\,.
 $$
 What do we know about $A$ and $B$? One can calculate the estimate for $A$, and $A=3$ works as we will see below. (In fact $A \le 2\sqrt{2}+o_d(1)$.)
 
 For the induction in \cite{KKL2017}  to close
 one need
 $$
 2d^{4-\frac{B}{16}} \le \frac12,
 $$
 Call $B$ which gives equality by $B_0$, $(4-B_0/16) \log d = -\log 4$, $B_0= 64 +\frac{16\log 4}{\log d}$.
 which forces $B\ge 64$ for sure. So $B> 2A^2$ is satisfied easily.
 
 The window, where are splitting to block estimate of $\bsa$ is better than the poly-log estimate is
 \begin{equation}
 \label{window}
 8B_0 < \frac{\log n}{d^2\log d} < \frac{4B_0^2}{9}\,.
 \end{equation}
 
 If we choose
 $$
 n=e^{8B_0 d^2\log d}
 $$
 then $Q_{pl} = n^{1/4}$, $Q_{block} \le n^{0.133}$.

\appendix

\section{Technical estimates for the block splitting bound}
\label{app:weak-calcs}

We record the deferred verifications from the proof of Theorem~\ref{thm:weak-est}.
Recall the notation
\[
F(n,a):=\maxbsa(d,n,a),
\qquad
A(n):=(K\log n)^{-d},
\qquad
b(n):=e^{\sqrt{\log n}},
\]
\[
\tau(n):=\frac{1}{(K\log n)^{16d} b(n)^{4d + 4\sqrt{2}\,d\sqrt{\log d}}},
\qquad
P(n):=\tau(n)^{-1/2}\Bigl(d\log \frac1{\tau(n)} \cdot \log n\Bigr)^{O(d)},
\]
and
\[
\Phi(n):=e^{M\sqrt{\log n}},
\]
where
\begin{equation}
\label{Md1}
M=M_d=\textcolor{red}{3}\,d\sqrt{\log d}\,.
\end{equation}

We work with sufficiently large $d$.

\subsection{Initial step of the induction}

We first verify that
\[
F(n,a)\le a\,\Phi(n)
\]
for all
\[
2\le n \le e^d, \,\, a\in  [0,1]
\]
Indeed, fix $2\le n<e^{d}$ and $a\in[0,1]$.
If $a\le A(n)$, then Lemma~\ref{lem:small alpha} gives
\[
F(n,a)\le a\le a\,\Phi(n).
\]
Assume now that $a>A(n)$.
By the trivial bound $\bsa[f]\le \sqrt n$ we have
\[
F(n,a)\le \sqrt n.
\]
Thus it suffices to show that
\[
\sqrt n\le A(n)\Phi(n)=
(K\log n)^{-d}e^{M_d\sqrt{\log n}}.
\]
For $2\le n\le e^d$, we have
\[
\log\frac{A(n)\Phi(n)}{\sqrt n}
\ge\frac{M_d}{2}\sqrt{\log n}-d\log(K\log n)\ge0
\]
for sufficiently large $d$, since $\log(K\log n)=O(\sqrt{\log n})$ for $n\ge2$, whereas $M_d/d\to\infty$.

\bigskip

Now we will verify that

\[
F(n,a)\le a\,\Phi(n)
\]
for
\[
e^d\le n<e^{M_d^2}
\qquad\text{and}\qquad
a\in[0,1],
\]
with $M_d=\textcolor{red}{3}d\sqrt{\log d}$ as above.
Indeed, fix $e^d \le n<e^{M_d^2}$ and $a\in[0,1]$.
If $a\le A(n)$, then Lemma~\ref{lem:small alpha} gives
\[
F(n,a)\le a\le a\,\Phi(n).
\]
Assume now that $a>A(n)$.
By the trivial bound $\bsa[f]\le \sqrt n$ we have
\[
F(n,a)\le \sqrt n.
\]
Thus it suffices to show that
\[
\sqrt n\le A(n)\Phi(n).
\]
Write
\[
x:=\sqrt{\log n}.
\]
Since $e^d\le n<e^{M_d^2}$, we have
\[
\sqrt{d}\le x<M_d.
\]
Hence
\begin{align*}
\log\!\Bigl(\frac{A(n)\Phi(n)}{\sqrt n}\Bigr)
&= M_d x-\frac{x^2}{2}-\textcolor{red}{d}\log(Kx^2)\\
&\ge \frac{M_d x}{2}-\textcolor{red}{d}\log(KM_d^2)\\
&\ge \frac{M_d\sqrt{d}}{2}-\textcolor{red}{d}\log(KM_d^2).
\end{align*}
For our choice of $M_d$ and sufficiently large $d$, we have
\[
\frac{M_d\sqrt{d}}{2}\ge d\,\log(KM_d^2),
\]
and hence
\[
A(n)\Phi(n)\ge \sqrt n.
\]
Therefore
\[
F(n,a)\le \sqrt n\le A(n)\Phi(n)\le a\,\Phi(n),
\]
which proves the initial step.

\subsection{Estimate for $\textnormal{II}$}

Suppressing integer rounding, we only need to prove
\begin{equation}
2\,C_d\,\Phi(n/b(n))
\le \frac12\Phi(n).
\end{equation}

Recall that, for an absolute constant $B$,
\begin{equation}
\label{Cd1}
2C_d\le Bd^4.
\end{equation}
It is therefore enough to show
\begin{equation}\label{eq:II-app}
Bd^4\,\Phi(n/b(n))\le \frac12\Phi(n).
\end{equation}
Write
\[
x:=\log n,
\qquad
b(n)=e^{\sqrt{x}},
\]
so that
\[
\log\frac{n}{b(n)}=x-\sqrt{x}.
\]
Then \eqref{eq:II-app} becomes
\[
Bd^4\,e^{M\sqrt{x-\sqrt{x}}}
\le \frac12 e^{M\sqrt{x}},
\]
or equivalently,
\[
2Bd^4
\le e^{M(\sqrt{x}-\sqrt{x-\sqrt{x}})}.
\]
Now
\[
\sqrt{x}-\sqrt{x-\sqrt{x}}
=
\frac{\sqrt{x}}
{\sqrt{x}+\sqrt{x-\sqrt{x}}}
\ge \frac12.
\]
Thus it suffices that
\[
2\,Bd^4\le e^{M/2}.
\]
For our choice $M=M_d=3d\sqrt{\log d}$,
this holds for sufficiently large $d$.
Consequently,
\[
\textnormal{II}\le \frac12 a\,\Phi(n).
\]

\subsection{Estimate for $\textnormal{I}$}

First,
\[
\tau(n)^{-1/2}
=
(K\log n)^{8d}e^{2d\sqrt{\log n}}e^{2\sqrt2\,d\sqrt{\log d}\sqrt{\log n}}.
\]
Also,
\[
\begin{aligned}
\log \frac1{\tau(n)}
&=16d\log(K\log n)+(4d+4\sqrt2\,d\sqrt{\log d})\sqrt{\log n}\\
&\le C\,d(\log(K\log n)+\sqrt{\log d}\sqrt{\log n}).
\end{aligned}
\]
Therefore, for absolute constants $C_0,c_0$ and $\log n\ge M_d^2$,
\[
P(n)
\le
d^{C_0d}(K\log n)^{c_0d} e^{2d\sqrt{\log n}+\textcolor{red}{2\sqrt2}\,d\sqrt{\log d}\sqrt{\log n}}.
\]
So,
\[
\textnormal{I}
\le
d^{C_0d}(K\log n)^{c_0d}\,e^{2d\sqrt{\log n}+ 2\sqrt2\,d\sqrt{\log d}\sqrt{\log n}}+3+e^{\sqrt{\log n}}.
\]
Recall that in the induction step we assume $a>A(n)=(K\log n)^{-d}$.
Thus it is enough to prove
\[
\textnormal{I} \le \frac12 A(n)\Phi(n)
= \frac12(K\log n)^{-d}e^{M\sqrt{\log n}}.
\]
Equivalently, after multiplying by $(K\log n)^{d}$, it suffices to show
\[
\begin{aligned}
& d^{C_0d}(K\log n)^{(c_0+1)d}
e^{2d\sqrt{\log n}+2\sqrt2\,d\sqrt{\log d}\sqrt{\log n}}\\
&\qquad +(3+e^{\sqrt{\log n}})(K\log n)^{d}
\le\frac12e^{M\sqrt{\log n}}.
\end{aligned}
\]
The extra $2d\sqrt{\log n}$ and the logarithms of the polynomial factors are $o_d(d\sqrt{\log d}\sqrt{\log n})$, uniformly for $\log n\ge M_d^2$. Since $2\sqrt2<3$, the displayed inequality follows for sufficiently large $d$.

Hence, with $M=M_d=3d\sqrt{\log d}$ and sufficiently large $d$, for $\log n\textcolor{red}{\ge}M_d^2$ we have
\[
\textnormal{I} \le \frac12 a\,\Phi(n).
\]
Together with the estimate for II and the initial range, this closes the induction for sufficiently large $d$.

\raggedbottom
\section{Application: Boundary geometry from $\bsa$}
\label{boundary}

The total influence of a Boolean-valued function counts the fraction of hypercube edges on the boundary between $A:=f^{-1}(-1)$ and $A^c=f^{-1}(1)$.
For two functions $f,g$ with the same total influence, their $\bsa$ values may differ significantly, and this variation reveals information about their \emph{vertex} boundary.
This is not too surprising, as $\bsa$ is the $1/2$-moment of the pointwise sensitivity:
for a fixed Boolean function $f$, the quantity $\sens{f}(x)$ is the number of sensitive edges attached to the vertex $x$.
Then
\[
\E_{x\sim\{0,1\}^n}[\sqrt{\sens{f}(x)}]=\bsa[f].
\]
Together with $\Inf[f]$ this is not enough to determine the \emph{size} of the vertex boundary
\[
\partial_{\mathrm{vert}}f:=\#\{x:\sens{f}(x)>0\}=2^n\Pr_{x\sim\{0,1\}^n}[\sens{f}(x)>0],
\]
but it does give some partial information.

For example,
\[
\var\!\bigl(\sqrt{\sens{f}(x)}\bigr)=\Inf[f]-\bsa[f]^2,
\]
so holding $\Inf[f]$ fixed, functions with smaller $\bsa$ have much more variance in their vertex sensitivities.
Another interpretation is as follows.

\begin{prop}
\label{prop:vertex-bdry}
Consider choosing a uniformly random edge $e$ from the boundary between $A$ and $A^{c}$, then from $e$ choosing either incident vertex $x$ with probability $1/2$.
Then
\begin{equation}
\label{edge1}
\Pr_{e\sim \partial A,\,x\sim e}\left[\sens{f}(x)\geq\frac{\Inf[f]^2}{4 \bsa[f]^2}\right]\geq \frac12.
\end{equation}
\end{prop}

A ``typical'' edge will thus be incident to a highly sensitive vertex when $\bsa$ is small in comparison with $\Inf$.
One may think about the special case of $\mathrm{MAJ}_n$ vs. $\chi_{\{1,\ldots, \sqrt n\}}$; see Fig.~\ref{fig:majvspar}.

\begin{samepage}
\begin{proof}[Proof of Proposition~\ref{prop:vertex-bdry}]
One computes:
\begin{align*}
\Pr_{\substack{e\sim \partial A\\x\sim e}}[\sens{f}(x)\le T]
&=\frac{\sum_x \sens{f}(x)\mathbf 1_{\{\sens{f}(x)\leq T\}}}{\sum_x \sens{f}(x)}
=\frac{\mathbf{E}\!\left[\sens{f}(x)\,\mathbf{1}_{\{\sens{f}(x)\le T\}}\right]}{\mathbf{E}\,\sens{f}(x)}\\
&\le \frac{\mathbf{E}\!\left[\sqrt{T}\,\sqrt{\sens{f}(x)}\,\mathbf{1}_{\{\sens{f}(x)\le T\}}\right]}{\mathbf{E}\,\sens{f}(x)}
\le \frac{\sqrt{T}\,\mathbf{E}\sqrt{\sens{f}(x)}}{\mathbf{E}\,\sens{f}(x)}\\
&=\frac{\sqrt{T}\,\bsa[f]}{\Inf[f]}\,.
\end{align*}
Substituting for $T$ completes the proof.
\end{proof}
\end{samepage}

\begin{figure}
    \centering
    \makebox[\textwidth][c]{\includegraphics[width=0.95\textwidth]{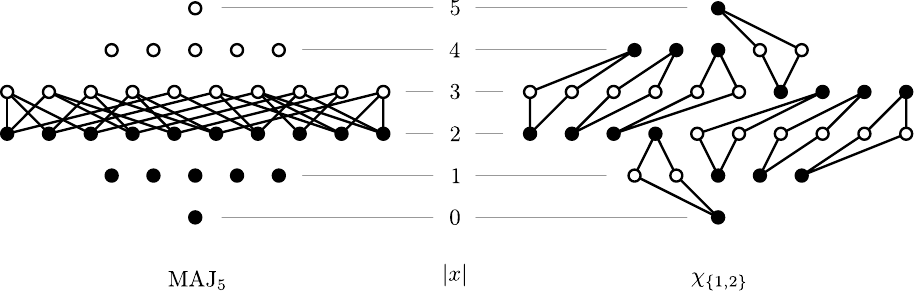}}
    \caption{The boundary of $\mathrm{MAJ}_5$ (left) vs. $\chi_{\{1,\ldots,\lfloor\sqrt 5\rfloor\}}$ (right).
    Points in $\{0,1\}^5$ are arranged according to Hamming weight, and $f(x)=1$ is denoted by $\bullet$ and $-1$ by $\circ$.
    Here $\Inf[\mathrm{MAJ}_5]=1.875$ (generally $\Theta(\sqrt n)$) and $\bsa[\mathrm{MAJ}_5]=5\sqrt{3}/8\approx 1.08$ (generally $\Theta(1)$), while $\Inf[\chi_{\{1,\ldots,\lfloor\sqrt 5\rfloor\}}]=2$ (generally $\Theta(\sqrt n)$) and $\bsa[\chi_{\{1,\ldots,\lfloor\sqrt 5\rfloor\}}]=\sqrt{2}$ (generally $\Theta(n^{1/4})$).
    So although $\mathrm{MAJ}_n$ and $\chi_{\{1,\ldots,\sqrt{n}\}}$ have comparable average sensitivities, the boundary vertices of $\mathrm{MAJ}_n$ will typically be more sensitive because $\bsa$ is small.}
    \label{fig:majvspar}
\end{figure}

In view of these remarks, our bounds on $\bsa$ for PTFs show that PTF vertex boundaries are small and highly sensitive.
Said another way: for PTFs, most inputs are very robust to perturbations or errors, while a small fraction of inputs are extremely sensitive to errors.

\medskip

Another application is superseded by the solution of Gotsman--Linial conjecture by OpenAI \cite{GLsol}. But maybe it is still worthwhile to mention.
Namely, we see by  our $polylog(n)$ estimate of $\bsa[f]$  and a naive estimate $\as \le \sqrt{n}\,\bsa$  prove Kane's estimate of $\as[f]$ in \cite{DK}.

\section{Application: Edge isoperimetric inequality for PTFs}
\label{edge}

A celebrated edge-isoperimetric inequality for Boolean functions by Kahn--Kalai--Linial (KKL) (we formulate it here for $f$ such that $\var[f]=1$), see \textcolor{red}{\cite{ODopen}}, page 267,
claims that
\[
\textnormal{MaxInf}[f] \ge \frac{9}{(\Inf[f])^2} 9^{-\Inf[f]}\,.
\]
From the Eldan--Gross result \cite{EG} one can also give the following variant
\[
\textnormal{MaxInf}[f] \ge \frac{c}{\Inf[f]} e^{- C(\bsa[f])^2}\,.
\]
For PTF's this is a much better estimate than the previous one.

\medskip

DISCLAIMER. The key ideas of this paper and also absolute majority of calculations are human generated, a few discussion with Open AI 6.0 Astra helped us to clarify  the role of different estimates.


\begin{thebibliography}{KKL2017}

\bibitem[Cha18]{C}
B.~Chapman.
\newblock The Gotsman--Linial conjecture is false.
\newblock {\em Proceedings of the Twenty-Ninth Annual ACM-SIAM Symposium on
  Discrete Algorithms, SIAM, Philadelphia, PA}, pages 692--699, 2018.

\bibitem[DHKM+10]{tan}
I.~Diakonikolas, P.~Harsha, A.~Klivans, R.~Meka, P.~Raghavendra, R.~A. Servedio,
and L.~Y. Tan.
\newblock Bounding the average sensitivity and noise sensitivity of polynomial threshold functions.
\newblock {\em Proceedings of the Forty-Second ACM Symposium on Theory of Computing}, pages 533--542, 2010.

\bibitem[DIRX26]{DIRX}
Polona Durcik, Paata Ivanisvili, Joris Roos, and Xinyuan Xie.
\newblock Sharp isoperimetric inequalities on the {H}amming cube {II}: The critical exponent.
\newblock {\em arXiv preprint}, arXiv:2602.20462, 2026.

\bibitem[EG22]{EG}
Ronen Eldan and Renan Gross.
\newblock Concentration on the {B}oolean hypercube via pathwise stochastic analysis.
\newblock {\em Invent. Math.}, 230(3):935--994, 2022.

\bibitem[EKLM25]{EKLM}
Ronen Eldan, Guy Kindler, Noam Lifshitz, and Dor Minzer.
\newblock Isoperimetric inequalities made simpler.
\newblock {\em Discrete Anal.}, pages Paper No.~7, 23, 2025.

\bibitem[IVHV20]{IVHV}
P.~Ivanisvili, R.~Van Handel, and A.~Volberg.
\newblock Rademacher type and {E}nflo type coincide.
\newblock {\em Ann. of Math. (2)}, 192(2):665--678, 2020.

\bibitem[Kan11]{DKgauss}
Daniel~M. Kane.
\newblock The Gaussian surface area and noise sensitivity of degree-$d$ polynomial threshold functions.
\newblock {\em Comput. Complexity}, 20(2):389--412, 2011.

\bibitem[Kan14]{DK}
Daniel~M. Kane.
\newblock The correct exponent for the {G}otsman--{L}inial conjecture.
\newblock {\em Comput. Complexity}, 23(2):151--175, 2014.

\bibitem[KKL17]{KKL2017}
V.~Kabanets, D.~M.~Kane, and Z.~Lu.
\newblock A polynomial restriction lemma with applications.
\newblock {\em Proceedings of the 49th Annual ACM SIGACT Symposium on Theory of Computing}, pages 615--628, 2017.

\bibitem[KOS08]{KODS}
Adam~R. Klivans, Ryan O'Donnell, and Rocco~A. Servedio.
\newblock Learning geometric concepts via Gaussian surface area.
\newblock In {\em 49th Annual IEEE Symposium on Foundations of Computer Science}, pages 541--550. IEEE, 2008.

\bibitem[Mar74]{M}
G.~A. Margulis.
\newblock Probabilistic characteristics of graphs with large connectivity.
\newblock {\em Problemy Peredachi Informatsii}, 10(2):101--108, 1974.

\bibitem[O'D12]{ODopen}
Ryan O'Donnell.
\newblock Open problems in analysis of Boolean functions.
\newblock {\em arXiv preprint}, arXiv:1204.6447, 2012.

\bibitem[OpenAI]{GLsol}
OpenAI.
\newblock{Average sensitivity of polynomial threshold functions}.
\newblock{September 25, 2026}

\bibitem[Per20]{YP}
Yuval Peres.
\newblock Noise stability of weighted majority.
\newblock In {\em In and out of equilibrium 3. Celebrating Vladas Sidoravicius}, volume~77 of {\em Progr. Probab.}, pages 677--682. Birkhauser/Springer, Cham, [2020] \copyright 2021.

\bibitem[FSVZ]{FSVZ}
F. Chang, J. Slote, A. Volberg, H. Zhang,
\newblock The boolean surface area of polynomial threshold function
\newblock 	arXiv:2604.08095, pp. 1--18.

\bibitem[Tal93]{T}
M.~Talagrand.
\newblock Isoperimetry, logarithmic {S}obolev inequalities on the discrete cube, and {M}argulis' graph connectivity theorem.
\newblock {\em Geom. Funct. Anal.}, 3(3):295--314, 1993.

\end{thebibliography}
\end{document}